\documentclass[aps,pra,superscriptaddress,twocolumn,showpacs]{revtex4}
\usepackage{amsmath,amsthm}
\usepackage{graphicx}
\usepackage{amssymb}
\usepackage{bbm, color}

\newcommand*{\cl}[1]{{\mathcal{#1}}}
\newcommand*{\bb}[1]{{\mathbb{#1}}}
\newcommand{\ket}[1]{\left|#1\right>}
\newcommand{\bra}[1]{\left<#1\right|}

\newcommand{\proj}[2]{| #1 \rangle\!\langle #2 |}
\newcommand*{\tn}[1]{{\textnormal{#1}}}

\newcommand{\T}{\mbox{$\textnormal{tr}$}}

\newtheorem{Prop}{Proposition}

\begin{document}

\title{Minimal control power of the controlled teleportation}

\author{Kabgyun Jeong}
\affiliation{School of Computational Sciences, Korea Institute for Advanced Study, Seoul 02455, Korea}
\affiliation{Center for Macroscopic Quantum Control, Department of Physics and Astronomy, Seoul National University, Seoul 08826, Korea}
\author{Jaewan Kim}
\affiliation{School of Computational Sciences, Korea Institute for Advanced Study, Seoul 02455, Korea}
\author{Soojoon Lee}
\affiliation{Department of Mathematics and Research Institute for Basic Sciences, Kyung Hee University, Seoul 02447, Korea}

\date{\today}

\begin{abstract}
We generalize the control power of a perfect controlled teleportation
of an entangled three-qubit pure state,
suggested by Li and Ghose [Phys. Rev. A {\bf 90}, 052305 (2014)],
to the control power of a general controlled teleportation
of a multiqubit pure state.
Thus, we define the minimal control power,
and calculate the values of the minimal control power
for a class of general three-qubit Greenberger-Horne-Zeilinger (GHZ) states
and the three-qubit $W$ class whose states have zero three-tangles.
Moreover, we show that the standard three-qubit GHZ state 
and the standard three-qubit $W$ state have
the maximal values of the minimal control power for the two classes, respectively.
This means that the minimal control power can be interpreted as
not only an operational quantity of a three-qubit quantum communication
but also a degree of three-qubit entanglement.
In addition, we calculate the values of the minimal control power for 
general $n$-qubit GHZ states and the $n$-qubit $W$-type states.
\end{abstract}

\maketitle

%%%
%%% Intro
%%%
\section{Introduction} \label{intro}
Quantum teleportation~\cite{BBCJPW93} has been considered
as one of the most important applications of quantum entanglement,
and hence has been studied in various ways,
which are experimental as well as theoretical~\cite{BPMEWZ97,FSBFKP98,BBMHP98,KKS01}.

In the standard quantum teleportation scheme,
the sender, Alice, and the receiver, Bob,
share a maximally entangled Bell state in advance,
and then Alice performs a two-qubit Bell measurement
on her qubit of the Bell state and a qubit state to be teleported.
Based on the two-bit classical information transmitted
according to the measurement outcome from Alice,
Bob can apply an appropriate unitary operation on his qubit of the Bell state
to perfectly recover the state.

We can also consider a teleportation scheme over a three-qubit pure state,
called {\em controlled teleportation} (CT)~\cite{KB98,LJK05,LJK07},
which is a variant of the splitting and reconstruction of quantum information
over the Greenberger-Horne-Zeilinger (GHZ) state,
introduced by Hillery {\it et al.}~\cite{HBB99}.
In this scheme,
a controller on one qubit of the state
maximally assists the teleportation procedure on the other two qubits.

Recently, Li and Ghose~\cite{LG14} investigated a {\em control power} of the controller
in perfect CT via two classes of the partially entangled three-qubit pure states,
and generalized it to multiqubit CT schemes~\cite{LG15}.
In this paper, we define a {\em minimal} control power of the CT,
which is a more general concept compared to the perfect case in Refs.~\cite{LG14,LG15},
and we present explicit calculations for the minimal control power for
a class of general three-qubit GHZ states
and the three-qubit $W$ class whose states have zero three-tangles.
Moreover, we show that the standard GHZ state and the standard $W$ state have
the maximal values of the minimal control power for the two classes, respectively.
This implies that our newly defined quantity,
the minimal control power,
can be naturally considered as a good candidate for a degree of three-qubit entanglement.
By using additional examples of general $n$-qubit GHZ states and the $n$-qubit $W$-type states,
we strengthen our claim on the degree of tripartite entanglement.

In Sec.~\ref{MCP},
we review the basic notations and known results,
and give a definition of the minimal control power.
In Sec.~\ref{Cal_MCP},
we present how to calculate the value of the minimal control power
for a three-qubit pure state.
In Secs.~\ref{CT/GHZ} and \ref{CT/W},
we explicitly calculate the minimal control power for the two classes of three-qubit pure states,
and show that the minimal control powers of the GHZ state and the $W$ state
are maximal in the two classes, respectively. Furthermore, in Sec.~\ref{NQCT},
we exhibit examples of general $n$-qubit GHZ states and the $n$-qubit $W$-type states.
Finally, we discuss our results in Sec.~\ref{conclusion}.

%%%
%%% Minimal Control Power
%%%
\section{Minimal Control Power in a CT}\label{MCP}
We first review the properties of
the maximal teleportation fidelity~\cite{P94}
and the fully entangled fraction~\cite{HHH99,BHHH00}.
The teleportation fidelity is naturally defined as
\begin{equation}
F(\rho)=\int d\xi \bra{\xi}\Lambda_{\rho}(\ket{\xi}\bra{\xi})\ket{\xi},
\label{eq:teleportation_fidelity}
\end{equation}
where $\Lambda_\rho$ is the standard teleportation protocol over a two-qubit state $\rho$
to attain the maximal fidelity,
the integral is performed
with respect to the uniform distribution $d\xi$ over all one-qubit pure states,
and the fully entangled fraction of $\rho$ is
defined as
\begin{equation}
f(\rho)=\max\bra{e}\rho\ket{e},
\label{eq:FEF}
\end{equation}
where the maximum is over all two-qubit maximally entangled states $\ket{e}$.
It has been shown~\cite{HHH99,BHHH00} that
\begin{equation}
F(\rho)=\frac{2f(\rho)+1}{3}.
\label{eq:2relation}
\end{equation}
We remark that $F(\rho)>2/3$ [or $f(\rho)>1/2$] if and only if
$\rho$ is said to be useful for teleportation,
since it has been shown that the classical teleportation
can have at most $F=2/3$ (or $f=1/2$)~\cite{BHHH00}.

Let $\rho_{12\cdots n}$ be an $n$-qubit state.
For an $(n-2)$-element subset $J\subseteq\{1,2,\ldots,n\}$ with $J=\{j_1,j_2,\ldots,j_{n-2}\}$,
let $F^{J}_\mathrm{CT}$ be
the maximal teleportation fidelity of the resulting two-qubit state in the subsystem~$kl$
after the measurement of the subsystem~$J$; that is,
\begin{equation}
F^{J}_\mathrm{CT}
=\max_{U_J}\sum_{t=0}^{2^{n-2}-1}\bra{t}U_J\rho_JU_J^\dagger\ket{t}F(\varrho_{kl}^t),
\label{eq:fij}
\end{equation}
where $U_J=U_{j_1}\otimes U_{j_2}\otimes\cdots\otimes U_{j_{n-2}}$,
$\rho_J=\rho_{j_1}\otimes \rho_{j_2}\otimes\cdots\otimes \rho_{j_{n-2}}$,
$\{k,l\}=\{1,2,\ldots,n\} \setminus J$,
the maximum is taken over all $2\times 2$ unitary operators $U_{j_1},U_{j_2},\ldots, U_{j_{n-2}}$,
and $\varrho_{kl}^{t}$ is the resulting state of the subsystem~$kl$
after the local measurement on the subsystem~$J$
when the measurement outcome is $t$.
Similarly, for all $J\subseteq\{1,2,\ldots,n\}$ with $|J|=n-2$,
let $f^{J}_\mathrm{CT}$ be defined as
the maximal average of the fully entangled fraction of the state in the subsystem~$kl$
after the measurement on the subsystem~$J$;
then it can be obtained that for all $J$
\begin{equation}
F^{J}_\mathrm{CT}=\frac{2f^J_\mathrm{CT}+1}{3},
%\label{eq:Fi}
\end{equation}
as in the two-qubit case~\cite{LJK05}.

We now note that $F_\mathrm{CT}^J\ge F(\rho_{kl})$ and $f_\mathrm{CT}^J\ge f(\rho_{kl})$.
Thus, we can define a control power as the difference between
the (controlled) maximal teleportation fidelity
and the teleportation fidelity without control
in a CT of an $n$-qubit state as follows.
For all $J\subseteq\{1,2,\ldots,n\}$,
let $P^J$ be the control power of an $n$-qubit state $\rho_{12\cdots n}$,
defined as
\begin{equation}
P^J(\rho_{12\cdots n}):=F_\mathrm{CT}^J-F(\rho_{kl}),
\end{equation}
where $P^J(\rho_{12\cdots n})\ge0$ due to the above note,
and let us define the minimal control power $P$ of an $n$-qubit state $\rho_{12\cdots n}$
by
\begin{equation} %\label{minimalCP}
P(\rho_{12\cdots n}):=\min_J\{P^J(\rho_{12\cdots n})\}.
\end{equation}
Then, by using the usefulness of teleportation on a two-qubit state,
we say that
the minimal control power $P$ is {\em meaningful}
on an $n$-qubit state $\rho_{12\cdots n}$
if and only if $F_\mathrm{CT}^J>2/3$ but $F(\rho_{kl})\le 2/3$
for all $J$, $k$, and $l$ satisfying $J\cup\{k,l\}=\{1,2,\ldots,n\}$.

In this paper, we deal with only pure states for explicit calculation,
since it is not easy to obtain a general formula of the minimal control power for mixed states
even though it can be calculated for a given mixed state.

%%%
%%% 3-qubit states
%%%
\section{Minimal Control Power in a Three-qubit CT} \label{MCP3}
Let $\psi_{123}=\proj{\psi}{\psi}_{123}$ be a three-qubit pure state.
For $j$ in $\{1,2,3\}$,
let $F^{j}_\mathrm{CT}$ be
the maximal teleportation fidelity of the resulting two-qubit state in the subsystem~$kl$
after the measurement on the system~$j$; that is,
\begin{equation}
F^{j}_\mathrm{CT}
=\max_{U}\left[\bra{0}U\rho_{j}U^{\dagger}\ket{0}F(\varrho_{kl}^{0})
+\bra{1}U\rho_{j}U^{\dagger}\ket{1}F(\varrho_{kl}^{1})\right],
\label{eq:fi}
\end{equation}
where the maximum is taken over all $2\times 2$ unitary matrices,
for $j\in\{1,2,3\}$, and
let $f^{j}_\mathrm{CT}$ be defined as
the maximal average of the fully entangled fraction of the state in the subsystem~$kl$
after the measurement of the subsystem~$j$.
For each $j\in\{1,2,3\}$,
let $P^j$ be the control power of a three-qubit pure state $\psi_{123}$, defined as
\begin{equation}
P^j(\psi_{123}):=F_\mathrm{CT}^j-F(\rho_{kl}),
\end{equation}
and let us define the minimal control power $P$ of a three-qubit pure state $\psi_{123}$
by
\begin{equation} \label{minimalCP}
P(\psi_{123}):=\min\{P^1(\psi_{123}), P^2(\psi_{123}), P^3(\psi_{123})\}.
\end{equation}

%%%
%%% How to calculate
%%%
\subsection{How to calculate the minimal control power}\label{Cal_MCP}
For a three-qubit pure state $\psi_{123}$,
the three-tangle $\tau$~\cite{CKW00,DVC} is defined as
\begin{equation}
\tau=\mathcal{C}^2_{j(kl)}-\mathcal{C}^2_{jk}-\mathcal{C}^2_{jl},
\label{eq:tangle}
\end{equation}
where $\mathcal{C}_{jk}=\mathcal{C}(\rho_{jk})=\mathcal{C}(\mathrm{tr}_l(\psi_{123}))$,
$\mathcal{C}_{j(kl)}=\mathcal{C}(\psi_{j(kl)})$,
and $\mathcal{C}$ is the Wootters' concurrence~\cite{HW,Wootters},
and the partial tangle $\tau_{jk}$~\cite{LJK05} is defined as
\begin{equation}
\tau_{jk}=\sqrt{\cl{C}_{j(kl)}^2-\cl{C}_{jl}^2}=\sqrt{\tau+\cl{C}_{jk}^2},
\label{eq:PT}
\end{equation}
for $\{j,k,l\}=\{1,2,3\}$.
Then it was shown~\cite{LJK05} that
there is an interesting relation between the maximal teleportation fidelity $F_\mathrm{CT}^j$
and the partial tangle $\tau_{kl}$, that is,
\begin{equation}
F_\mathrm{CT}^j=\frac{2+\tau_{kl}}{3}.
\label{eq:MTF_PT}
\end{equation}
Since the Wootters' concurrence for a two-qubit state is computable,
each maximal teleportation fidelity $F_\mathrm{CT}^j$ for a three-qubit pure state
is also computable.

For distinct $j$, $k$, and $l\in\{1,2,3\}$,
let $T^j$ be a $3\times 3$ real matrix
whose $(m,n)$ entry is $\T\left(\rho_{kl}(\sigma_m\otimes\sigma_n)\right)$,
where $\sigma_i$'s are the Pauli matrices.
Then it was shown~\cite{HHH96} that
\begin{equation}
F(\rho_{kl})=\frac{3+\|T^j\|_1}{6},
\label{eq:TF_T}
\end{equation}
where $\|\cdot\|_1$ is the trace norm; that is,
$\|M\|_1=\T\sqrt{M^\dagger M}$ for any matrix $M$.
Thus, given a three-qubit pure state,
the teleportation fidelity without control $F(\rho_{kl})$ can be calculated,
and this directly implies that
one can calculate the minimal control power
for a given three-qubit pure state.

From Eqs.~(\ref{eq:MTF_PT}) and (\ref{eq:TF_T}),
we can know that
the minimal control power is meaningful
if and only if each partial tangle is strictly positive,
and each $\|T^j\|_1$ is not greater than 1; that is,
$\tau_{kl}>0$ but $\|T^j\|_1\le1$ for all $j$, $k$, and $l$.

%%%
%%% GHZ
%%%
\subsection{Example: General GHZ states}\label{CT/GHZ}
Let $\ket{\psi_{\tn{GHZ}}}$ be a state defined by
\begin{equation}
\ket{\psi_{\tn{GHZ}}}=a\ket{000}+b\ket{111},
\end{equation}
where $a,b\in\bb{C}$ such that $|a|^2+|b|^2=1$; then
the state $\ket{\psi_{\tn{GHZ}}}$ is here called a general GHZ state,
and let $\psi_{\tn{GHZ}}=\proj{\psi_{\tn{GHZ}}}{\psi_{\tn{GHZ}}}_{123}$.
Then it is clear that its reduced density matrices have the same form as
\begin{equation}
\rho_{kl}
=\begin{pmatrix}
|a|^2 & 0 & 0 & 0 \\
0 & 0 & 0 & 0 \\
0 & 0 & 0 & 0\\
0 & 0 & 0 & |b|^2
\end{pmatrix},
\end{equation}
and the matrices $T^j$ also have the same form,
\begin{align}
T^j&
=\begin{pmatrix}
0 & 0 & 0 \\
0 & 0 & 0 \\
0 & 0 & 1
\end{pmatrix}.
\end{align}
Since $\|T^j\|_1=1$,
it can be directly obtained that the teleportation fidelity $F(\rho_{kl})=2/3$
for all distinct $k,l\in\{1,2,3\}$.
In addition, we can readily know that
the three-tangle $\tau(\psi_{\tn{GHZ}})=4|a|^2|b|^2$ and the concurrence $\cl{C}_{kl}=0$
since $\rho_{kl}$ is separable.
Hence, the maximal teleportation fidelity of a general GHZ state becomes
\begin{equation}
F_\mathrm{CT}^j=\frac{\sqrt{\tau+\cl{C}_{kl}^2}+2}{3}=\frac{2|a||b|+2}{3}.
\label{eq:MTP_GHZ}
\end{equation}

Since each control power for the state
\begin{equation}
P^j(\psi_{\tn{GHZ}})=F_\mathrm{CT}^j(\psi_{\tn{GHZ}})-F(\rho_{kl})=\frac{2|a||b|}{3},
\end{equation}
the minimal control power is
\begin{equation}
P(\psi_{\tn{GHZ}})=\frac{2|a||b|}{3}.
\label{eq:MCP_GHZ}
\end{equation}

We note that
the minimal control power $P$ is meaningful for the state $\psi_{\tn{GHZ}}$
if both $a$ and $b$ are non-zero,
as seen in Eq.~(\ref{eq:MTP_GHZ}).
Since the inequality $|a||b|\le 1/2$ holds
for all complex numbers $a$ and $b$ satisfying $|a|^2+|b|^2=1$,
and the equality in the relation holds if and only if $|a|=|b|=1/\sqrt{2}$,
it is clearly shown that the standard GHZ state attains the maximal value of the minimal control power among the general GHZ states.

%%%
%%% W
%%%
\subsection{Example: $W$-class states}\label{CT/W}
Let $\psi_{{W}}=\ket{\psi_{{W}}}\bra{\psi_{{W}}}$ be a $W$-class state,
whose three-tangle vanishes, that is, $\tau(\psi_{{W}})=0$.
Then it is known~\cite{AACJLT00,ABLS01} that the state $\psi_{{W}}$ can be written as
\begin{equation}
\ket{\psi_{{W}}}=\lambda_0\ket{100}+\lambda_1\ket{000}+\lambda_2\ket{110}+\lambda_3\ket{101}
\end{equation}
up to local unitary,
where the coefficients $\lambda_i\ge0$ and $\sum_i\lambda_i^2=1$.

By straightforward calculations,
we can find the matrices $T^j$ as follows:
\begin{eqnarray}
T^1&=&\begin{pmatrix}
2\lambda_2\lambda_3 & 0 & 2\lambda_0\lambda_2 \\
0 & 2\lambda_2\lambda_3 & 0 \\
2\lambda_0\lambda_3 & 0 & 1-2(\lambda_2^2+\lambda_3^2)
\end{pmatrix}, \\
T^2&=&\begin{pmatrix}
2\lambda_1\lambda_3 & 0 & 2\lambda_0\lambda_1 \\
0 & -2\lambda_1\lambda_3 & 0 \\
-2\lambda_0\lambda_3 & 0 & 1-2(\lambda_0^2+\lambda_2^2)
\end{pmatrix}, \\
T^3&=&\begin{pmatrix}
2\lambda_1\lambda_2 & 0 & 2\lambda_0\lambda_1 \\
0 & -2\lambda_1\lambda_2 & 0 \\
-2\lambda_0\lambda_2 & 0 & 1-2(\lambda_0^2+\lambda_3^2)
\end{pmatrix}.
\end{eqnarray}
Hence, for each distinct $j$, $k$, and $l$ in $\{1,2,3\}$,
the trace norm is given by
\begin{equation}
\|T^j\|_1=2\lambda_k\lambda_l+\sqrt{A_j},
\end{equation}
where
\begin{equation}
A_j=\max\left\{
%(1+2\lambda_k\lambda_l)^2-4\lambda_j(\lambda_k+\lambda_l)^2,
\big(\lambda_0^2+(\mp\lambda_j\pm\lambda_k+\lambda_l)^2\big)
\big(\lambda_0^2+(\lambda_j+\lambda_k\pm\lambda_l)^2\big)
\right\}.
%\big(\lambda_0^2+(\lambda_j+\lambda_k-\lambda_l)^2\big)\big(\lambda_0^2+(\lambda_j-\lambda_k+\lambda_l)^2\big).
\end{equation}
Thus we obtain that
\begin{equation}
F(\rho_{kl})=\frac{2\lambda_k\lambda_l+\sqrt{A_j}+3}{6}
\label{eq:Fkl_3qubits}
\end{equation}
for all distinct $j$, $k$, and $l$.

Since we can see that $\cl{C}_{kl}=2\lambda_k\lambda_l$,
each maximal teleportation fidelity $F_\mathrm{CT}^j$ can be obtained as follows:
\begin{equation}
F_\mathrm{CT}^j(\psi_{{W}})=\frac{\cl{C}_{kl}+2}{3}=\frac{2\lambda_k\lambda_l+2}{3}.
\label{eq:F_CTj}
\end{equation}
Thus, since for each distinct $j$, $k$, and $l$ in $\{1,2,3\}$
the control power becomes
\begin{equation}
P^j(\psi_{{W}})=\frac{1}{6}\left(2\lambda_k\lambda_l+1-\sqrt{A_j}\right),
\end{equation}
the minimal control power for a $W$-class state is
\begin{equation}
P(\psi_{{W}})
=\min\left\{\frac{1}{6}\left(2\lambda_k\lambda_l+1-\sqrt{A_j}\right)\right\},
\end{equation}
where the minimum is taken over all distinct $j$, $k$, and $l$ in $\{1,2,3\}$.

From the fact that $\lambda_j^2(\lambda_k-\lambda_l)^2\ge0$,
we can show the inequality $\sqrt{A_j}\le1-2\lambda_k\lambda_l$.
This implies that
\begin{equation}
F(\rho_{kl})=\frac{\|T^j\|_1+3}{6}\le \frac{2}{3},
\label{eq:F_W2}
\end{equation} for all distinct $j$, $k$, and $l$.
Hence, from Eq.~(\ref{eq:F_CTj}) and the inequality~(\ref{eq:F_W2}),
we can say that the minimal control power is meaningful for any $W$-class state
if and only if all $\lambda_j$'s are non-zero.

Let ${W}$ be the standard W state, that is,
${W}=\psi_{{W}}$ with
$\lambda_0=0$ and $\lambda_1=\lambda_2=\lambda_3={1}/{\sqrt{3}}$.
Then we can obtain the following proposition.
\begin{Prop}\label{Prop}
For any $W$-class state $\psi_{{W}}$,
\begin{equation}
P(\psi_{{W}})\le\frac{2}{9}=P({W}).
\end{equation}
\end{Prop}

\begin{proof}
Suppose that there exists a $W$-class state $\psi_{{W}}$
such that $P(\psi_{{W}})>\frac{2}{9}$, that is,
\begin{equation}
\frac{1}{6}\left(2\lambda_k\lambda_l+1-\sqrt{A_j}\right)>\frac{2}{9}
\label{eq:A_j}
\end{equation}
for all distinct $j$, $k$, $l$ in $\{1,2,3\}$.
For each $j\in\{1,2,3\}$,
let $A_j'$ be defined as
\begin{equation}
A_j':=\max\{
(\mp\lambda_j\pm\lambda_k+\lambda_l)^2
(\lambda_j+\lambda_k\pm\lambda_l)^2
\},
\label{eq:A_j'}
\end{equation}
where $j$, $k$, and $l$ are all distinct in $\{1,2,3\}$,
then it is clear that $A_j'\le A_j$ for all $j\in\{1,2,3\}$.
It follows from the inequality in (\ref{eq:A_j}) that,
for all distinct $j$, $k$, and $l$ in $\{1,2,3\}$,
\begin{equation}
\frac{1}{6}\left(2\lambda_k\lambda_l+1-\sqrt{A'_j}\right)>\frac{2}{9}.
\label{eq:A_r'}
\end{equation}
Hence, we have
\begin{equation}
\frac{1}{6}\left(2\lambda_k\lambda_l+1
-|\mp\lambda_j\pm\lambda_k+\lambda_l||\lambda_j+\lambda_k\pm\lambda_l|\right)
>\frac{2}{9}
\label{maxl}
\end{equation}
for all distinct $j$, $k$, and $l$ in $\{1,2,3\}$.

Without loss of generality,
we may now assume that $\lambda_j=\max\{\lambda_1,\lambda_2,\lambda_3\}$.
Then the inequality~(\ref{maxl}) becomes
\begin{equation}
\frac{1}{2}\left(2\lambda_k\lambda_l+1
-(\lambda_j+\lambda_k-\lambda_l)(\lambda_j-\lambda_k+\lambda_l)\right)
>\frac{2}{3},
\end{equation}
which is equivalent to the inequality
\begin{equation}
1-\lambda_j^2+\lambda_k^2+\lambda_l^2>\frac{4}{3}.
\label{eq:43}
\end{equation}
Since $1-\lambda_j^2\ge \lambda_k^2+\lambda_l^2$,
we can obtain from the inequality~(\ref{eq:43}) that
$\lambda_j^2<{1}/{3}$,
which is a contradiction due to the assumption that $\lambda_j$ is maximal.
\end{proof}

%%%
%%% N-qubit states
%%%
\section{Minimal Control Power in an $n$-qubit CT} \label{NQCT}
\subsection{$n$-qubit GHZ states}\label{CT/NGHZ}
Let $|\psi_{\tn{GHZ}}^{(n)}\rangle$ be an $n$-qubit GHZ state defined by
\begin{equation}
|\psi_{\tn{GHZ}}^{(n)}\rangle=a\ket{00\cdots0}+b\ket{11\cdots1},
\end{equation}
where $a,b\in\bb{C}$ such that $|a|^2+|b|^2=1$.
Then, as seen in the three-qubit case in Sec.~\ref{CT/GHZ},
it is clear that $F(\rho_{kl})={2}/{3}$ for any distinct $k$ and $l$ in $\{1,2, \ldots, n\}$, 
and it is also clear that $F_{\tn{CT}}^J=2(|a||b|+1)/3$ for any $(n-2)$-element subset $J$ of $\{1,2, \ldots, n\}$. 
Thus we have the (minimal) control power for $n$-qubit GHZ states
\begin{equation}
P^J(\psi_{\tn{GHZ}}^{(n)})=P(\psi_{\tn{GHZ}}^{(n)})=\frac{2|a||b|}{3}
\end{equation}
for all $(n-2)$-element subsets $J$ of $\{1,2, \ldots, n\}$,
and it is totally equivalent to the case of general three-qubit GHZ states in Sec.~\ref{MCP3}.

\subsection{$n$-qubit $W$-type states}\label{CT/NW}
Let $\psi_{{W}}^{(n)}=|\psi_{{W}}^{(n)}\rangle\langle\psi_{{W}}^{(n)}|$ be
an $n$-qubit $W$-type state, defined by
\begin{eqnarray}
|\psi_{{W}}^{(n)}\rangle
&=&\alpha_1\ket{100\cdots0}+\alpha_2\ket{010\cdots0} \nonumber\\
&&+\alpha_3\ket{0010\cdots0}+\cdots+\alpha_n\ket{00\cdots01},
\end{eqnarray}
where $\alpha_i\in\bb{C}$ such that $\sum_{i=1}^n|\alpha_i|^2=1$.
Then, for $J=\{1,2,\ldots,n\}\setminus\{k,l\}$ ($k>l$),
we have a reduced density matrix as
\begin{eqnarray}
\rho_{kl}^{{W}}
&=&\T_J|\psi_{{W}}^{(n)}\rangle\langle\psi_{{W}}^{(n)}| \nonumber\\
&=&\big(\alpha_k\ket{10}+\alpha_l\ket{01}\big)\big(\alpha_k^*\bra{10}
+\alpha_l^*\bra{01}\big) \nonumber\\
&&+\sum_{j\in J}|\alpha_j|^2
\proj{00}{00},
\end{eqnarray}
where it is known that $\cl{C}(\rho_{kl}^{{W}})=2|\alpha_k||\alpha_l|$
in Ref.~\cite{KS08}. If we define $\alpha=\sqrt{\sum_{j\in J}|\alpha_j|^2}$, then,
by using straightforward calculation,
\begin{equation}
\|T^J\|_1=4|\alpha_k||\alpha_l|+\big||\alpha|^2-|\alpha_k|^2-|\alpha_l|^2\big|.
\end{equation}
Thus, we have
\begin{equation}
F(\rho_{kl}^{{W}})
=\frac{3+4|\alpha_k||\alpha_l|+\big||\alpha|^2-|\alpha_k|^2-|\alpha_l|^2\big|}{6}
\le\frac{2}{3},
\end{equation}
since $|\alpha|^2=1-|\alpha_k|^2-|\alpha_l|^2$. Also we can obtain that
\begin{equation}
F_{\tn{CT}}^J(\psi_{{W}}^{(n)})
=\frac{2|\alpha_k||\alpha_l|+2}{3}>\frac{2}{3}.
\end{equation}
For this reason, we can say that the minimal control power is meaningful
for $n$-qubit $W$-type states if $\alpha_i\neq 0$ for all $i$.

Now we formulate the (minimal) control power for an $n$-qubit $W$-type state.
For any $J\subseteq\{1,2,\ldots,n\}$, the control power is given by
\begin{eqnarray}
P^J(\psi_{\tn{W}}^{(n)})
&=&\frac{1+\big|1-2\big(|\alpha_k|^2+|\alpha_l|^2\big)\big|}{6} \\
&=&\begin{cases}
\frac{|\alpha|^2}{3}&\text{if} \;\;|\alpha_k|^2+|\alpha_l|^2\le\frac{1}{2}, \\
\frac{|\alpha_k|^2+|\alpha_l|^2}{3}&\text{if}\;\;|\alpha_k|^2+|\alpha_l|^2>\frac{1}{2}.
\end{cases}
\end{eqnarray}
Thus, we have the minimal control power of those type:
\begin{equation}
P(\psi_{\tn{W}}^{(n)})=\frac{1}{6}\min\left\{1+\big|1-2\big(|\alpha_k|^2
+|\alpha_l|^2\big)\big|\right\}.
\end{equation}
Note that, for the $n$-qubit standard $W$ class
$|{W}^{(n)}\rangle=\frac{1}{\sqrt{n}}\ket{100\cdots0}+\cdots+\ket{00\cdots01}$,
the minimal control power is given by
\begin{equation}
P(\tn{W}^{(n)})=\begin{cases}
\frac{1}{3}-\frac{2}{3n}&\text{if} \;\;n\ge4 \\
\frac{2}{9}&\text{if}\;\;n=3.
\end{cases}
\end{equation}
By employing Proposition~\ref{Prop}, 
it can be shown that
\begin{equation}
P(\psi_W^{(n)}) \le \frac{2}{9} = P(W_3^{(n)}),
\label{eq:max_Wn}
\end{equation}
where $W_3^{(n)}$ is an $n$-qubit W-type state 
with $\alpha_{j_1}=\alpha_{j_2}=\alpha_{j_3}=1/\sqrt{3}$
and $\alpha_j = 0$ for all $j$ in $\{1,2, \ldots, n\}\setminus\{j_1,j_2,j_3\}$. 

\section{Conclusion}\label{conclusion}
We have considered the CT of a multiqubit state,
and have presented a new concept called the minimal control power
representing how faithfully the CT can be performed.
In addition, we have calculated the values of the minimal control power
for a class of general three-qubit GHZ states
and the three-qubit $W$ class whose states have zero three-tangles.
Extending the three-qubit cases, we also introduced a minimal control power of $n$-qubit
GHZ and $W$-type states.

Moreover, we have shown that the standard GHZ state and the standard $W$ state have
the maximal values of the minimal control power for the two classes.
We can also obtain a similar result for $n$-qubit GHZ and $W$-type states.
Therefore, this implies that
our new quantity, the minimal control power,
has not only an operational meaning in the CT
but also is an appropriate property
as a degree of tripartite entanglement.

For $W$-type states, their three-tangle values are zero,  
even though these states have three-party correlations. 
These correlations are apparently captured 
by the minimal control power discussed in our paper. 
In other words, the minimal control power could be 
a candidate for a measure of three-party correlation
complementary to the three-tangle.

\acknowledgements{
We are grateful to an anonymous referee for valuable comments.
This work was partly supported by the IT R\&D program of MOTIE/KEIT [10043464].
K.J. acknowledges financial support by the National Research Foundation of Korea (NRF) Grant funded by the Korea government (MSIP) (Grant No. 2010-0018295). 
S.L. acknowledges financial support
by the Basic Science Research Program through the National Research Foundation of Korea
funded by the Ministry of Education (NRF-2012R1A1A2003441)
and the Associate Member Program funded by the Korea Institute for Advanced Study.
}


\begin{thebibliography}{26}

\bibitem{BBCJPW93} C.~H.~Bennett, G.~ Brassard, C.~Cr\'{e}peau, R.~Jozsa, A.~Peres, and W.~K.~Wootters,
%{\em Teleporting an unknown quantum state via dual classical and Einstein-Podolsky-Rosen channels},
\prl~{\bf 70}, 1895 (1993).
%
\bibitem{BPMEWZ97} D.~Bouwmeester, J.-W.~Pan, K.~Mattle, M.~Eibl, H.~Weinfurter, and A.~Zeilinger,
%{\em Experimental quantum teleportation},
Nature (London)~{\bf 390}, 575 (1997).
%
\bibitem{FSBFKP98} A.~Furusawa, J.~L.~S{\o}rensen, S.~L.~Braunstein, C.~A.~Fuchs, H.~J.~Kimble, and E.~S.~Polzik,
%{\em Unconditional Quantum Teleportation},
Science~{\bf 282}, 706 (1998).
%
\bibitem{BBMHP98} D.~Boschi, S.~Branca, F.~De~Martini, L.~Hardy, and S.~Popescu,
%{\em Experimental Realization of Teleporting an Unknown Pure Quantum State via Dual Classical and Einstein-Podolsky-Rosen Channels},
\prl~{\bf 80}, 1121 (1998).
%
\bibitem{KKS01} Y.-H.~Kim, S.~P.~Kulik, and Y.~Shih,
%{\em Quantum Teleportation of a Polarization State with a Complete Bell State Measurement},
\prl~{\bf 86} 1370 (2001).
%
\bibitem{KB98} A.~Karlsson and M.~Bourennane,
%{\em Quantum teleportation using three-particle entanglement},
\pra~{\bf 58}, 4394 (1998).
%
\bibitem{LJK05} S.~Lee, J.~Joo, and J.~Kim,
%{\em Entanglement of three-qubit pure states in terms of teleportation capability},
\pra~{\bf 72}, 024302 (2005).
%
\bibitem{LJK07} S.~Lee, J.~Joo, and J.~Kim,
%{\em Teleportation capability, distillability, and nonlocality on three-qubit states},
\pra~{\bf 76}, 012311 (2007).
%
\bibitem{HBB99} M.~Hillery, V.~Bu\v{z}ek, and A.~Berthiaume,
%{\em Quantum secret sharing},
\pra~{\bf 59}, 1829 (1999).
%
\bibitem{LG14} X.-H.~Li and S.~Ghose,
%{\em Control power in perfect controlled teleportation via partially entangled channels},
\pra~{\bf 90}, 052305 (2014).
%
\bibitem{LG15} X.-H.~Li and S.~Ghose,
%{\em Analysis of $N$-qubit perfect controlled teleportation schemes from the controller's point of view},
\pra~{\bf 91}, 012320 (2015).
%
\bibitem{P94} S.~Popescu,
%{\em Bell's Inequalities versus Teleportation: What is Nonlocality?},
\prl~{\bf 72}, 797 (1994).
%
\bibitem{HHH99} M.~Horodecki, P.~Horodecki, and R.~Horodecki,
%{\em General teleportation channel, singlet fraction, and quasidistillation},
\pra~{\bf 60}, 1888 (1999).
%
\bibitem{BHHH00} P.~Badziag, M.~Horodecki, P.~Horodecki, and R.~Horodecki,
%{\em Local environment can enhance fidelity of quantum teleportation},
\pra~{\bf 62}, 012311 (2000).
%
\bibitem{CKW00} V.~Coffman, J.~Kundu, and W.~K.~Wootters,
%{\em Distributed entanglement},
\pra~{\bf 61}, 052306 (2000).
%
\bibitem{DVC} W.~D\"{u}r, G.~Vidal, and J.~I.~Cirac,
%{\em Three qubits can be entangled in two inequivalent ways},
\pra~{\bf 62}, 062314 (2000).
%
\bibitem{HW} S.~Hill and W.~K.~Wootters,
%{\em Entanglement of a Pair of Quantum Bits},
\prl~{\bf 78}, 5022 (1997).
%
\bibitem{Wootters} W.~K.~Wootters,
%{\em Entanglement of Formation of an Arbitrary State of Two Qubits},
\prl~{\bf 80}, 2245 (1998).
%
\bibitem{HHH96} R.~Horodecki, M.~Horodecki, and P.~Horodecki,
%{\em Teleportation, Bell's inequalities and inseparability},
Phys. Lett. A~{\bf 222}, 21 (1996).
%
\bibitem{AACJLT00} A.~Ac\'{i}n, A.~Andrianov, L.~Costa, E.~Jan\'{e}, J.~I.~Latorre, and R.~Tarrach,
%{\em Generalized Schmidt Decomposition and Classification of Three-Quantum-Bit States},
\prl~{\bf 85}, 1560 (2000).
%
\bibitem{ABLS01} A.~Ac\'{i}n, D.~Bruss, M.~Lewenstein, A.~Sanpera,
%{\em Classification of Mixed Three-Qubit States},
\prl~{\bf 87}, 040401 (2001).
%
\bibitem{KS08} J.~S.~Kim and B.~C.~Sanders,
%{\em Generalized W-class state and its monogamy relation},
J. Phys. A: Math. Theor.~{\bf41}, 495301 (2008).



\end{thebibliography}
\end{document}